\documentclass[10pt]{article}                    
\usepackage{graphicx}
\usepackage{amsmath}
\usepackage{amsthm}
\usepackage{amsfonts}
\usepackage{amssymb}
\usepackage{graphicx}
\usepackage{lipsum}

\newtheorem{thm}{Theorem}[]
\newtheorem{lem}[thm]{Lemma}

\newtheorem*{defn*}{Definition}
\newtheorem*{rem*}{Remark}
\newtheorem{prop}[thm]{Proposition}
\newtheorem{cor}[thm]{Corollary}
\newtheorem*{claim*}{Claim}

\newcommand{\Prb}{\mathbb{P}}
\newcommand{\Bin}[2]{\left(\begin{array}{c} #1 \\ #2 \end{array}\right)}

\newcommand{\Z}{\mathbb{Z}}
\newcommand{\R}{\mathbb{R}}

\begin{document}

\title{Rigorous Confidence Intervals on Critical Thresholds in 3 Dimensions
}


\author{Neville Ball}




\maketitle

\begin{abstract}
We extend the method of Balister, Bollob\'{a}s and Walters {[\ref{MW}]} for determining rigorous confidence intervals for the critical threshold of two dimensional lattices to three (and higher) dimensional lattices. We describe a method for determining a full confidence interval and apply it to show that the critical threshold for bond percolation on the simple cubic lattice is between $0.2485$ and $0.2490$ with $99.9999\%$ confidence, and the critical threshold for site percolation on the same lattice is between $0.3110$ and $0.3118$ with $99.9999\%$ confidence.
\end{abstract}

\section{Introduction}
We say that a graph that has a drawing that corresponds to a regular tiling of $\R^{d}$ is an (infinite) $d$-dimensional lattice. For any such lattice $\Lambda$, we define the independent bond percolation model on $\Lambda$ by assigning each edge of $\Lambda$ a state, either open or closed, independently with probability $p$. Percolation theory then seeks to answer the question: `when is there an infinite open component (i.e. an infinite collection of vertices connected via open edges)?' Kolmogorov's 0-1 law tells us that there exists a critical threshold $p_{c}$ such that for $p<p_{c}$ the probability of having an infinite component is $0$, while if $p>p_{c}$ the probability of having at least one infinite open component is $1$ (see e.g. Chapters~1 and 3 of Grimmet [\ref{GG}]).

However, apart from a handful of simple two dimensional lattices, the exact value of $p_{c}$ is not known, and indeed for many lattices it may not be possible to know the value of $p_{c}$ exactly, as there is no reason to believe that they are algebraic numbers, or even easily expressible. As a result, work has been put into estimating critical probabilities. Until recently this was typically done in one of two ways, either as a fully rigorous upper or lower bound, or as a heuristic estimate based on computer simulations.

The former generally give a fairly wide interval, e.g. the best known bounds for site percolation on the square lattice, even after a lot of work (e.g. Toth [\ref{Toth}], Luczak and Wierman [\ref{LuAndWie}] and Zuev [\ref{Zuev}]) are $p_{c}\geq 0.556$, due to van den Berg and Ermakov [\ref{BergAndErm}], and $p_{c}\leq 0.679492$, due to Wierman [\ref{Wierman}].

The heuristic estimates, on the other hand, claim a very high level of accuracy, e.g. Newman and Ziff, in [\ref{Z'}], claim the critical threshold for site percolation on the simple cubic lattice is $0.59274621\pm 1.3\times 10^{-8}$. However, it is hard to be sure of how accurate these estimates really are, particularly since most of the more accurate estimates assume unproven results about scaling theory e.g. see Ziff and Neumann [\ref{Z2}].

In [\ref{BS}] Bollob\'{a}s and Stacey (in the context of oriented percolation) and then Balister, Bollob\'{a}s and Walters in [\ref{MW}] (for continuum percolation) introduced a method which is intermediate between the heuristic estimates and the rigorous bounds. They were able to produce a rigorous confidence interval for critical thresholds in two dimensions, i.e. they were able to rigorously prove that with a certain confidence the critical threshold lay within a certain (narrow) interval. E.g. Using this method Riordan and Walters proved in [\ref{MW}] that with $99.9999\%$ confidence the critical threshold for site percolation on the square lattice was between $0.5925$ and $0.5930$. However the method they used does not generalise to lattices in any dimension higher than two.

In Section~2 we describe the method of Balister, Bollob\'{a}s and Walters for finding a rigorous confidence interval for a two dimensional lattice, in Section~3 we extend the method to three and higher dimensions and apply it in Section~4 to prove:

\begin{thm}
On the simple cubic lattice:
\begin{itemize}
\item $[0.2485,0.2490]$ is a $99.9999\%$ confidence interval for the critical threshold in bond percolation.
\item $[0.3110,0.3118]$ is a $99.9999\%$ confidence interval for the critical threshold in site percolation.
\end{itemize}
\end{thm}

In Section~5 we discuss higher dimensional lattices.

\section{Rigorous confidence intervals in two dimensions}

The idea behind the method of Balister, Bollob\'{a}s and Walters in [\ref{MW}] is to reduce percolation on an infinite two-dimensional lattice, $\Lambda$, to something that can be tested on a finite number (although large) number of sites. The number of sites was too large to produce an exact result, however, they were able to gain a confidence interval via repeated simulation.

More precisely, they first reduce the problem to that of a 1-independent percolation problem on $\Z^{2}$. In 1-independent percolation we only insist that states of the bonds in two collections of edges are independent if the graph distance between every edge in the first collection and every edge in the second collection is at least one (we generally think of this condition as insisting that the states of vertex disjoint edges being independent, although we do need the stronger condition as stated above). Balister, Bollob\'{a}s and Walters, [\ref{MW}], proved the following lemma:

\begin{lem}\label{Lem1}
Given a 1-independent bond percolation measure on $\Z^{2}$ in which each bond is open with a probability at least $p_{0}=0.8639$,  the probability that the origin lies in an infinite open cluster is positive. Thus, in particular, percolation will occur.
\end{lem}

The exact reduction is as follows:

We choose a scale parameter $s>0$, and tile $\R^{2}$ with $s$ by $s$ squares. We identify these squares with $\Z^{2}$, and so label them as $S_{v}$ for $v\in \Z^{2}$. If $u$ and $v$ are two adjacent points of $\Z^{2}$, then we define the rectangle $R_{uv}=S_{u}\cup S_{v}$. Let $E_{uv}$ be an event that depends only on the state of the bonds in $R_{uv}$ and declare $uv$ to be open if $E_{uv}$ holds and closed otherwise. Since the rectangles corresponding to vertex disjoint edges of $\Z^{2}$ are disjoint this will define a 1-independent model on $\Z^{2}$.

Suppose we can choose an event $E$ and a value of $p$ such that:
\begin{enumerate}
\item An infinite open path in the resulting 1-independent model on $\Z^{2}$ implies an infinite open component in $\Lambda$.
\item The probability that $E$ holds (dependent on $p$) on any given rectangle is at least $0.8639$.
\end{enumerate}
Then Lemma~\ref{Lem1} tells us that the probability of the origin of $\Z^{2}$ being in an infinite open cluster is positive, and this implies that there exists an infinite open cluster in $\Lambda$, and $p$ must be above the percolation threshold.

In both [\ref{MW1}] and [\ref{MW}] the event used was as follows:
For each bond $uv\in\Z^{2}$, let $\Lambda_{u}$ and $\Lambda_{v}$ be the subgraphs induced by the sites in $S_{u}$ and $S_{v}$, and $\Lambda_{uv}$ be the subgraph induced by the sites in $R_{uv}$. Let $E_{uv}$ be the event that both $\Lambda_{u}$ and $\Lambda_{v}$ contain a unique largest open cluster, and that these large components are connected within $\Lambda_{uv}$.

Since $E_{uv}$ depends only on the states of the bonds in $R_{uv}$, and clearly fulfils property 1 of the above, to find an upper bound on the percolation threshold, it suffices to find a pair $(s,p)$ for which condition 2 holds.

If we could prove rigorously that condition 2 holds for some pair $(s,p)$, then we would achieve a rigorous bound on $p_{c}$. Unfortunately, for very small $s$ the bounds obtained are fairly weak, and for larger values of $s$ the event becomes complicated (e.g. all of the simulations done by Balister and Walters in [\ref{MW1}] considered blocks containing over $10^{9}$ sites). However, since $s$ is finite, we have reduced our infinite problem (on all of $\Lambda$) to one on a finite block of $\Lambda$ (and thus checkable). Now, if we generate a block of $\Lambda_{uv}$ at random, then there is some fixed probability, $p_{1}$, such that $E_{uv}$ holds. Thus, if we generate $N$ such blocks at random, the number $n$ for which $E_{uv}$ holds will be distributed binomially with parameters $N$ and $p_{1}$. Namely,
\begin{align}
\Prb(n\geq m)=\sum_{i=m}^{N}\Bin{N}{i}p_{1}(1-p_{1})^{N-i}\label{Bin1}
\end{align}

For any given confidence level $\alpha$, we wish to generate a one sided confidence interval for $p_{c}$ with confidence $\alpha$ (i.e. a random value $X$ such that $\Prb(X>p_{c})=\alpha$). We can do this by generating a confidence interval for the bond probability, $p'$, that ensures that $\Prb(E_{uv})>0.8639$ for some (fixed) scale factor $s$ as follows:

We fix a pair $(s,p)$ and run the simulation some large fixed number, $N$ say, of times, with a scale factor of $s$ and a bond probability of $p$. If we get a high number of successes (i.e. simulations where the event $E_{uv}$ occurs), say at least $m$ successes, then we set $X=p$ and otherwise set $X=1$. Then we can only have $X\leq p'$ if $p\leq p'$, in which case the success probability of a given trial must be at most $0.8639$. Thus, by (\ref{Bin1}):
\begin{align}
\Prb(X\leq p')=&\Prb(\text{We have at least }m\text{ successes}\notag\\
&\,\,\text{given that the success probability}\notag\\
&\,\,\text{is at most }0.8639)\notag\\
\leq&\sum_{i=m}^{N}\Bin{N}{i}0.8639^{i}(1-0.8639)^{N-i}\label{Bin2}
\end{align}
If we have chosen $m$ such that the right hand side of (\ref{Bin2}) is less than $1-\alpha$, then we would have that $\Prb(X>p')\geq\alpha$, and so $\Prb(X>p_{c})\geq\alpha$. Thus, making one evaluation of $X$, which is what we report, we will have a one sided confidence interval of the desired significance for $p_{c}$. Note that if we choose our value of $p$ to be too low (or are just unlucky), then we will only be able to report the trivial upper bound of $1$.

To make this confidence interval into a two sided confidence interval Balister, Bollob\'{a}s and Walters used planar duality: Defining the dual lattice $\Lambda^{*}$ of $\Lambda$ to be the lattice with a vertex corresponding to each face of $\Lambda$, and a bond $e^{*}$ for each bond $e$, which joins the two sites of $\Lambda^{*}$ in which $e$ lies. Then, for any two dimensional planar lattice with rotational symmetry at least two (Bollob\'{a}s and Riordan [\ref{OBR}]):
\begin{align}
p_{c}+p^{*}_{c}= 1\label{Planarity}
\end{align}
Thus to bound $p_{c}$ below, it is enough to bound $p_{c}^{*}$ above. In fact, (\ref{Planarity}) has been proved more generally without rotational symmetry conditions by Sheffield in [\ref{SS}], although all of the lattices studied in [\ref{MW}] met the symmetry conditions required in [\ref{OBR}].

To gain an upper bound on $p_{c}$, this method generalises to higher dimensions `as is.' We will be considering $\Lambda$ as a subset of $\R^{d}$ rather than $\R^{2}$, but Lemma~\ref{Lem1} will still hold, since an infinite open cluster on any two dimensional subset is still an infinite open cluster. We are thus looking at percolation on a slab, but since the percolation threshold of a slab is trivially at least that of the full lattice, and tends to the percolation threshold of the full lattice as the thickness of the slab increases, the upper bound we obtain will hold, and can be made as tight as required by considering large enough blocks.

Planar duality certainly does not hold in any dimension higher than two, and so this method will not give a lower bound. We examine this problem in the next section.

\section{Gaining a lower bound on the critical threshold in higher dimensions}\label{S1}

We aim to prove a lower bound on the percolation threshold of a lattice without using planar duality. Ideally we would be able to mirror the method of Balister, Bollob\'{a}s and Walters outlined in the last section, making only the following two changes:
\begin{enumerate}
\item We need to change Lemma~\ref{Lem1} to a lower bound result, i.e. something of the form `Given a 1-independent percolation measure on $\Z^{d}$ in which each bond is open with a probability $p<p_{0}$, the probability that the origin lies in an infinite open cluster is $0$.'
\item We need to find an event $E$ dependent only on a finite block of $\Lambda$ such that an infinite open cluster in $\Lambda$ will give an infinite open cluster in our new lattice over $\Z^{d}$.
\end{enumerate}

However, there is a problem with this approach: It is possible (if unlikely) that we could have an infinite component in $\Lambda$ which only ever had very short paths in any given block (e.g. suppose we had an infinite path which zig-zagged back and forth across the boundary between two blocks (See Figure~1 for an illustration of this), or worse, in higher dimensions, we could have a path which circled around a line which was on the boundary of more than two blocks, so that it would create no long path even in any set of two adjacent blocks). This makes it difficult to come up with an event which is dependent only on the states of bonds contained entirely within a given block but still meets condition 2 in the above.

\begin{figure}[ht]
\centering
\includegraphics[width=80mm]{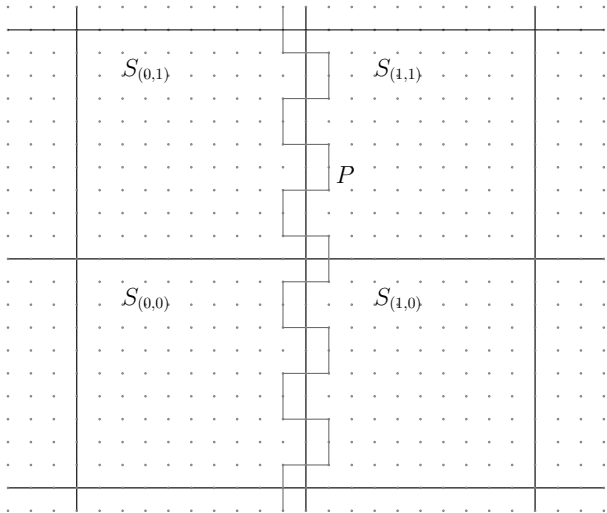}
\caption{A section of an infinite path through four blocks that creates no section of long path in any given block.}
\end{figure}

It is worth noting that there are different ways of gaining lower bounds: Given a vertex $x\in\Lambda$, we define $B_{r}(x)$ to be the set of vertices within a graph distance $r$ of $x$ (in $\Lambda$), $S_{r}(x)$ to be those sites that are exactly a graph distance $r$ from $x$ and $N_{r}(x)$ to be the number of sites in $S_{r}(x)$ that can be reached by open paths within $B_{r}(x)$. Then a simple expectation argument shows that if $\mathbb{E}_{p}[N_{r}(x)]<1$ for any $r>0$ then we must have $p<p_{c}$. Whilst it is possible to estimate the value of $\mathbb{E}_{p}[N_{r}(x)]$ and still be able to give rigorous error bounds (and thus gain a similar confidence style result), the bounds will tend to be quite weak, in part because it is hard to bound the expectation of a random variable that can take very large values occasionally.

Instead, we present a way around the problems with extending the method of Balister, Bollob\'{a}s and Walters to higher dimensions: We are going to allow our blocks to overlap. In fact, Bollob\'{a}s and Riordan had the seed of this method outlined in Chapter~4 of [\ref{BR}], but did not prove a version of Lemma~\ref{Lem1} that was independent of the block size taken. 

Allowing the blocks to overlap has two consequences:
\begin{itemize}
\item The state of edges within blocks near one another will no longer be independent from each other.
\item The resulting graph on the blocks may no longer be a simple copy of $\Z^{d}$ (depending on how we let them overlap).
\end{itemize}
Both of these points make it harder to get a version of Lemma~\ref{Lem1} with a good lower bound (i.e. a value of $p_{0}$ which isn't too low). This is important for two reasons; firstly a weaker upper bound will mean we have to take a larger block size (i.e. scale factor) to ensure that the probability of our event occurring is low enough, and secondly we will have to do more trials to gain a given level of confidence. If our bound becomes too low then the computing power required to gain a meaningful result would become impractical. However, in three dimensions, everything is still manageable.

For simplicity, from now on we will assume our original lattice $\Lambda$ is a copy of the simple cubic lattice on $\Z^{3}$ and concentrate on bond percolation, although everything that follows works for site percolation too, as well as a large class of three dimensional lattices, and indeed higher dimensional lattices (See Section~\ref{SecHD}).

We define the probability measure $\Prb_{p}$ on $\Lambda$ by taking every edge of $\Lambda$ to be open independently with probability $p$ and closed otherwise. We will allow our blocks to overlap as follows:

\begin{defn*}
We tile $\R^{3}$ by disjoint copies of the unit cube. We then associate with each unit cube a cube of side length $2$ centred on the same point and with the same orientation (i.e. the faces of the two cubes are parallel), and call each cube of side length $2$ a \emph{block}.

For a given block, we define the \emph{centre} of that block to be the original unit cube centred on the same point, and say that two blocks are adjacent if their centres share a face. We call the resulting structure the \emph{block lattice} and denote it $\beta$. Note that as a graph with these adjacency rules $\beta$ is isomorphic to $\Z^{3}$.

We define a \emph{path of blocks} to be a sequence of distinct adjacent blocks, $B_{0}$, $B_{1}\ldots$, and call such a path \emph{minimal} if $B_{i}$ and $B_{j}$ are adjacent if and only if $|i-j|=1$.
\end{defn*}

We wish to consider $\Lambda$ as being contained within the blocks: Given any positive integer $n$ (which we will assume is a multiple of four), we embed $\Lambda$ into $\R^{3}$ such that: Adjacent vertices of $\Lambda$ are a distance $\frac{1}{n}$ apart, the orientation of $\Lambda$ lines up with the orientation of the block lattice, and the nearest vertices of $\Lambda$ to the boundary of each block are all a distance $\frac{1}{2n}$ from the boundary of that block. This ensures that every vertex of $\Lambda$ is inside the centre of a unique block, and no vertex is on the boundary of any block.

\begin{defn*}
We say that a vertex $v\in \Lambda$ is on the surface of a block if it is within that block, but there is a vertex $w\in\Lambda$ adjacent to $v$ which is not.
\end{defn*}

We can now define our event on a block as follows:

\begin{defn*}
Given a block $B$, we define $E_{B}$ to be the event that there exists an open path (in $\Lambda$) from a vertex on the surface of the block to any vertex $v$ contained within the centre of the block, and then another (vertex and edge) disjoint open path from $v$ to another vertex in the surface of the block.

Given a scale factor $n$ and a measure $\Prb_{p}$ on $\Lambda$, we define the probability measure $\tilde{\Prb}_{p}$ on $\beta$ by declaring any block $B$ to be open if $E_{B}$ occurs and closed otherwise.
\end{defn*}

\begin{prop}\label{Pr1}
The existence of an infinite open path, $P$, in $\Lambda$ implies the existence of an infinite minimal open path in the block lattice, $\beta$.
\end{prop}

\begin{proof}
Every vertex in an infinite open path in $\Lambda$ is inside the centre of a unique block. Since there are only finitely many points inside each block and the path is infinite, there must be two paths that are edge and vertex disjoint (outside of the centre) to vertices in the surface of the block (except for possibly the first two blocks whose centres contain points in $P$ if $P$ is only singly infinite). Thus every block in $\beta$ which has the infinite path in $\Lambda$ running through its centre (except for possible the first two) is open.

If an infinite path has a vertex in the centre of a block, it must also have a vertex in the centre of at least one adjacent block (e.g. the first vertex not in the centre of the current block). Thus there must be an infinite walk $B_{1}$, $B_{2}\ldots$ of open blocks, and each block can only appear finitely many times in this walk (since each grid block has only a finite number of vertices).

Now, any such infinite walk will contain a minimal infinite path by deleting all blocks between the first appearance of a block and the last appearance of any of it's neighbours over all blocks sequentially.
\end{proof}

It is worth noting that whilst we were considering a \emph{bond} percolation model on $\Lambda$, it is a \emph{site} percolation model that we consider on $\beta$.

We wish to prove a lower bound on the critical threshold of $\beta$ under $\tilde{\Prb}$, and to do this it is convenient to consider only blocks which are independent of each other. To this end we wish to construct a lattice on $\beta$ in which adjacent blocks are independent of each other:

Note that two blocks are open independently of each other if they do not overlap, which is the case if and only if their centres (as subsets of $\R^{3}$) share no faces, edges or corners. So, given an infinite minimal path, $P\in\beta$, and $B_{0}\in P$, the next block on $P$ which is open independently of $B_{0}$ must be one of the the blocks adjacent (in $\beta$) to the cube of blocks whose centres form a cube of side length 3 centred on the centre of $B_{0}$. This gives a total of 54 possible blocks (9 per face of the cube of side length 3).

\begin{defn*} We define the \emph{independence lattice}, denoted $\Upsilon$, to be the lattice on the same vertex set as $\beta$, but with edge set, from any given block $B$, exactly the 54 edges to the blocks which are adjacent to one of the blocks whose centre surrounds the centre of $B$.
\end{defn*}

\begin{prop}\label{Pr2}
The existence of an infinite open path in $\beta$ implies the existence of an infinite minimal open path in $\Upsilon$ such that every block in this path is open independently from all other blocks.
\end{prop}

\begin{proof}
Given an infinite open path $P_{0}\in\beta$, we can find an infinite minimal path, $P_{1}\in\beta$ as a subset of $P_{0}$. We induce an infinite path, $P_{2}$, in $\Upsilon$ by fixing a block $B_{0}\in P_{1}$ and then taking the last block in $P_{1}$ adjacent to this in $\Upsilon$ to be the next block in $P_{2}$. Since we picked the last possible block each time in $P_{1}$, the centres of the blocks in $P_{2}$ can never share a face, edge of corner, and so must be mutually independent. Finally $P_{2}$ contains an infinite minimal path by deleting all blocks between the first appearance of a block and the last appearance of any of it's neighbours over all blocks sequentially.
\end{proof}

Using the above proposition, we can now bound the critical threshold in $\beta$ under $\tilde{\Prb}$. It would be trivial to gain a lower bound of $\frac{1}{54}$, since if the probability of a block being open is less than $\frac{1}{54}$, then the expected number of paths with $n$ independent blocks tends to zero, and so the probability of having such a path must also tend to zero. However, we are trying to prove as strong a bound as possible, since it will give us a higher level of confidence not only for this lattice, but for other three dimensional lattices too. We thus strengthen this result to the following:

\begin{thm}\label{T1}
If $p'$ is the probability of a block being open in $\beta$, and $p'\leq\frac{3}{100}$, then the probability of having an infinite open cluster in $\beta$ is $0$.
\end{thm}

\begin{proof}
By Proposition~\ref{Pr2}, if we do not have any infinite minimal paths of blocks in $\Upsilon$ such that each block is independent from the others, we cannot have an infinite path in $\beta$, and so the probability of having an infinite open cluster in $\beta$ is $0$.

Now, we want to bound the number of such paths in $\Upsilon$ of a given length. Suppose we take a pair of adjacent vertices in $\Upsilon$ and consider how many minimal paths of a given length there are starting from this pair. This number can vary depending on the relationship between the two vertices in $\Upsilon$, and there are 3 distinct types of adjacent pair (the neighbours of any block, $B\in\Upsilon$, are the blocks that have their centres sharing a face with the $3\times 3$ cube centred on the centre of $B$ (denote this $3\times$ cube $C$), and each face of this cube is equivalent. Thus there are three sorts of neighbours, i.e. the neighbours can have their centre sharing a face with either the centre, edge or corner of a face of $C$).

Let the $3\times 3$ matrix $\textbf{M}_{k}$ have $(i,j)^{th}$ entry equal to the number of paths of length $k$ which start from a pair of neighbours of type $i$ and end in a pair of type $j$. Let $\textbf{n}_{n}$ be the $3$-vector of the number of paths of length $n$ ending in an edge of each type. We have:
\begin{equation}
\textbf{n}_{n+k} \leq \textbf{M}_{k} \textbf{n}_{n}
\end{equation}
We can use this to gain a recurrence relation bounding the number of paths of a given length ending in edges of each type. The solutions to these recurrence equations over paths of length $mk + c$ for $c \in \{1,2...k\}$ will be of the form;
\begin{equation}
\textbf{n}_{mk+c} \leq \sum_{i=1}^{3}{\textbf{c}_{i}e_{i}^{m}}
\end{equation}
Where the $e_{i}$ are the eigenvalues of $\textbf{M}_{k}$, and the $\textbf{c}_{i}$ are constant $3$-vectors dependent only on $c$.
Thus there exists a constant $C$, such that the total number of paths of length $n$ will be bounded by $C\cdot (E^{1/k})^{n}$, where $E$ is the largest eigenvalue of $\textbf{M}_{k}$.

Calculating $M_{k}$ for $k=6$, we obtain the following:
\[M_{6}=\left(\begin{array}{ccc}
139068488 & 147798994 & 145131436\\
708801255 & 754445397 & 740361638\\
438727951 & 465222047 & 455921413
\end{array}\right)\]

This has characteristic equation:
\begin{align}
f(\lambda) & = \lambda^{3}-1349435298\lambda^{2}-574193103868851\lambda\notag\\
&\, +212282708057868352770
\end{align}

Mathematica tells us that $f(\lambda)$ has it's largest root at $E=1.349860\ldots\times 10^{9}$ which has $6^{th}$ root equal to $33.244\ldots$ Thus if $p' <\frac{1}{33.244\ldots}=0.03008\ldots$, then we have that, for some constant $C$:
\begin{align}
 \mathbb{P}(\exists \textrm{ a path of length}\geq k)& \leq C\cdot (Ep')^{k}\notag\\
 & \rightarrow 0\textrm{ as }k\rightarrow \infty
\end{align}

To be fully rigorous, we show explicitly that $(100/3)^{6}$ is greater than the largest eigenvalue, and so that if $p' <3/100<1/33.244\ldots$ then we will have no such path in $\Upsilon$ with probability~$1$. We have:
\begin{align}
f(0) & = 212282708057868352770>0\notag
\end{align}
\begin{align}
f(250000) & = -15589649034344397230<0\notag
\end{align}
And:
\begin{align}
f\left(\left(\frac{100}{3}\right)^{6}\right) & = \frac{1.59\ldots\times 10^{34}}{387420489} > 0\notag
\end{align}
Thus $f(\lambda)$ has three real roots and it's largest is strictly less than $(100/3)^{6}$.
\end{proof}

\begin{cor}
Given a pair $(n,p)$, if the probability of a block in $\beta$ being open under $\tilde{\Prb}_{p}$ is less than $3/100$ then $p<p_{c}$, the critical threshold for $\Lambda$.
\end{cor}

\section{Results}
We wrote two programs, one to to test our lower bound event, and one to test our upper bound event (using the upper bound event of Balister, Bollob\'{a}s and Walters of [\ref{MW1}] described in Section~2) for both bond and site percolation on the simple cubic lattice. For the lower bound cases, we ran 800 trials for each lattice, noting that if we write $\phi(p)$ for the probability of a block being open for our lower bound event, and $X$ for the number of open blocks in our trial, we have:
\begin{align}
\Prb(&X\leq 4|\phi(p)>0.03)\notag\\
& \leq \sum_{i=0}^{4}\Bin{800}{i}\cdot0.03^{i}(1-0.03)^{800-i}\notag\\
& = 4.796\ldots\times10^{-7}\notag\\&<5\times10^{-7}\notag
\end{align}
For the upper bound cases we ran 400 trials for each lattice (since our bound for dependent percolation in the upper bound event (i.e. $0.8639$) is further from~$1$ than our bound in the lower bound case (i.e. $0.03$) is from~$0$, we need to do far less trials to achieve a similar level of confidence). Writing $\psi(p)$ for the probability of an edge being open for our upper bound event, and $Y$ for the number of edges found to be open, we have:
\begin{align}
\Prb(&Y\geq 378|\psi(p)<0.8639)\notag\\
& \leq \sum_{i=378}^{400}\Bin{400}{i}\cdot0.8639^{i}(1-0.8639)^{400-i}\notag\\
& \leq 1.1489\ldots\times10^{-7}\notag\\&<5\times10^{-7}\notag
\end{align}

Thus, if we have at most four trials in our lower bound simulations finding an open block, and at least 378 upper bound simulations finding an open edge, then we can say with $99.9999\%$ confidence that $p_{c}$ is above the value of $p$ used in our lower bound simulations, and below that used for our upper bound simulations.

To decide on the parameters to use, we ran our programs several times on varying different sizes of lattice and bond probabilities. We also used high precision heuristic estimates, such as in [\ref{Z}], as a guide (this claims that the percolation threshold for bond percolation on the simple cubic lattice is $0.2488126\pm0.0000005$). None of this affects our statistics so long as we do only one final run.

To generate the random numbers used to make instances of the lattices we used an SIMD-orientated fast Mersenne Twister (dSFMT ver 2.1 [\ref{Twist}]), and seeding each new trial with a new seed. To ensure that we could only do one final run, we decided the final trials would use consecutive seeds starting from:
\begin{itemize}
\item For bond percolation lower bound: $12345$.
\item For bond percolation upper bound: $123456$.
\item For site percolation lower bound: $1234567$.
\item For site percolation upper bound: $12345678$.
\end{itemize}

All pre-trial simulations would use seeds that wouldn't overlap with any of those above.

The final parameters used are listed in Table~\ref{table1}.

\begin{table}
\caption{The lattice sizes and bond and site probabilities used in our simulations}
\label{table1}
\begin{tabular}{llll}
\hline\noalign{\smallskip}
Percolation Type & Bound & Block Size & Bond/Site Probability\\
\noalign{\smallskip}\hline\noalign{\smallskip}
Bond & Lower & $6500^{3}$ & $0.2485$\\
Bond & Upper & $3000^{3}$ & $0.2490$\\
Site & Lower & $5000^{3}$ & $0.3110$\\
Site & Upper & $3000^{3}$ & $0.3118$\\
\noalign{\smallskip}\hline
\end{tabular}
\end{table}

The number of seeds which met our percolations events are listed in Table~\ref{table2}.

\begin{table}
\caption{The number of trials that met our percolation event for each simulation}
\label{table2}
\begin{tabular}{lll}
\hline\noalign{\smallskip}
Percolation Type & Bound & Trials meeting Percolation Event\\
\noalign{\smallskip}\hline\noalign{\smallskip}
Bond & Lower & 4\\
Bond & Upper & 400\\
Site & Lower & 4\\
Site & Upper & 397\\
\noalign{\smallskip}\hline
\end{tabular}
\end{table}

The seeds that met the percolation event in the lower bound cases, and did not meet the percolation event in the upper bound event (call these the \emph{bad seeds}) are listed in Table~\ref{table3}.

\begin{table}
\caption{The number of bad seeds for the random number generator for each simulation}
\label{table3}
\begin{tabular}{lll}
\hline\noalign{\smallskip}
Percolation Type & Bound & Bad Seeds\\
\noalign{\smallskip}\hline\noalign{\smallskip}
Bond & Lower & 12455, 13084, 13181, 13241\\
Bond & Upper & -\\
Site & Lower & 1235150, 1235186, 1235236, 1235274\\
Site & Upper & 12345787, 12345841, 12345991\\
\noalign{\smallskip}\hline
\end{tabular}
\end{table}

Thus we can conclude that with $99.9999\%$ confidence the percolation thresholds on the simple cubic lattice lie between $0.2485$ and $0.2490$ for bond percolation and between $0.3110$ and $0.3118$ for site percolation.

\section{Higher Dimensions}\label{SecHD}
As already mentioned, the method of Balister, Bollob\'{a}s and Walters generalises to any number of dimensions to give a confidence upper bound for the critical threshold. The method outlined in Section~\ref{S1} will in fact generalise to give a confidence lower bound:

If we change $\beta$ so that two blocks are adjacent if their centres share a face, edge or corner, then Proposition~\ref{Pr1} will hold for any $d$-dimensional lattice. We can then define $\Upsilon$ exactly as we did in Section~\ref{S1}, except now each block in $\Upsilon$ will have $5^{d}-3^{d}$ neighbours, and Proposition~\ref{Pr2} will hold by the same proof (if $\Lambda$ is such that it can be mapped onto a copy of $\Z^{d}$ in such a way that no edges go between blocks whose centres don't share a face, this can be reduced to $2d3^{d-1}$. This holds for every Archimidean lattice other than the triangular lattice in two dimensions, and so any lattice made by stacking Archimidean lattices (such as graphite). It also holds for the $\Z^{d}$ for any $d$).

A version of Theorem~\ref{T1} will then apply, but the bound will be exponentially small in $d$, and so the eventual confidence bound would be much weaker. However, since the probability of a component having a diameter at least $d$ decreases exponentially in $d$, if the bond probability, $p$, is below $p_{c}$, the probability of $E_{B}$ occurring will tend to $0$ as the block size increases, and so it would still be possible to achieve a confidence lower bound on $p_{c}$.

It would however, require exponentially more computing power to obtain similarly tight results, and so in higher dimensions, the level of confidence would inevitably be lower. E.g. in four dimensions each vertex in $\Upsilon$ has degree $216$, so we would need to show that the probability of a block of $\Upsilon$ being open is less than this. This could be done, but would take a large number of simulations. However in five dimensions each vertex would have degree $810$, at which point the computing power required becomes somewhat unrealistic (even if we had no bad seeds, we would still have to run at least 3682 trials to get a confidence level of $99\%$).




\end{document}